\documentclass[11pt]{article}

\usepackage{graphicx}
\usepackage{amsmath,amsthm}
\usepackage{amssymb}
\usepackage{fullpage}
\usepackage{enumerate}
\usepackage{graphicx}
\usepackage{hyperref}

\input{Qcircuit}

\def\ket#1{| #1 \rangle}
\def\bra#1{\langle #1 |}
\def\bracket#1#2{\langle #1 | #2 \rangle}

\newtheorem{theorem}{Theorem}
\newtheorem{lemma}[theorem]{Lemma}

\title{\bf Characterization of Binary Constraint System Games}

\author{Richard Cleve \thanks{Institute for Quantum Computing and School of Computer Science, University of Waterloo}
\and Rajat Mittal \thanks{Institute for Quantum Computing and Department of Combinatorics and Optimization, University of Waterloo}}

\date{}

\begin{document}

\maketitle 

\begin{abstract}
We consider a class of nonlocal games that are related to binary constraint systems (BCSs) in a manner similar to the games implicit in the work of Mermin 
[N.~D.~Mermin, ``Simple unified form for the major no-hidden-variables theorems," 
\textit{Phys.\ Rev.\ Lett.}, \textbf{65}(27):3373--3376, 1990], 
but generalized to $n$ binary variables and $m$ constraints.
We show that, whenever there is a perfect entangled protocol for such a game, there exists a set of binary observables with commutations and products similar to those exhibited by Mermin.
We also show how to derive upper bounds strictly below 1 for the the maximum entangled success probability of some BCS games.
These results are partial progress towards a larger project to determine the computational complexity of deciding whether a given instance of a BCS game admits a perfect entangled strategy or not.
\end{abstract}

%------------------------------------------------------------------------------%
\section{Binary constraint system games}
%------------------------------------------------------------------------------%

Constraint systems and various two-player non-local games associated with them have played an important role in both computational complexity theory (probabilistic interactive proof systems~\cite{Ben-OrG+1988,BabaiF+1991,FeigeG+1996,AroraL+1998} and the hardness of approximation~\cite{FeigeG+1996}) and quantum information (pertaining to the power of entanglement~\cite{Bell1964,ClauserH+1969,Tsirelson1980,CleveH+2004}).

A \textit{binary constraint system (BCS)} consists of $n$ binary variables, $v_1, v_2, \dots, v_n$, and $m$ constraints, $c_1, c_2, \dots, c_m$, where each $c_j$ is a binary-valued function of a subset of the variables.
For convenience, we may write the constraints as equations.
An example of a BCS (with $n=9$ and $m=6$) is 
\begin{eqnarray}\label{eq:ms}
v_1 \oplus v_2 \oplus v_3 & = & 0 
\ \ \ \ \ \ \ \ \ \ \ \ \ \ \ \ 
v_1 \oplus v_4 \oplus v_7 \ = \ 0 \nonumber \\
v_4 \oplus v_5 \oplus v_6 & = & 0 
\ \ \ \ \ \ \ \ \ \ \ \ \ \ \ \ 
v_2 \oplus v_5 \oplus v_8 \ = \ 0 \\
v_7 \oplus v_8 \oplus v_9 & = & 0 
\ \ \ \ \ \ \ \ \ \ \ \ \ \ \ \ 
v_3 \oplus v_6 \oplus v_9 \ = \ 1 \nonumber
\end{eqnarray}
(this BCS is related to the version of Bell's theorem introduced by Mermin \cite{Mermin1990}, that is discussed further in the next section).
If, as in this example, all the constraints are functions of the parity of a subset of variables we call the system a \textit{parity BCS}.
A BCS is \textit{satisfiable} if there exists a truth assignment to the variables that satisfies every constraint.
The above example is easily seen to be unsatisfiable (since summing all the equations modulo 2 yields $0 = 1$).

We can associate a two-player non-local game with each BCS that proceeds as follows.
There are two cooperating players, Alice and Bob, who cannot communicate with each other once the porotocol starts, and a verifier.
The verifier randomly (uniformly) selects one constraint $c_s$ and one variable $x_t$ from $c_s$.
The verifier sends $s$ to Alice and $t$ to Bob.
Alice returns a truth assignment to all variables in $c_s$ and Bob returns a truth assignment to variable $x_t$.
The verifier accepts the answer if and only if: 
\begin{enumerate}
\item Alice's truth assignment \textit{satisfies} the constraint $c_s$; 
\item Bob's truth assignment for $x_t$ is \textit{consistent} with Alice's.
\end{enumerate}
Strategies where Alice and Bob employ no entanglement are called \textit{classical}.
Strategies where they employ entanglement are called \textit{quantum} (or \textit{entangled}).
A strategy is \textit{perfect} if it always succeeds.

It is not too hard to see that there exists a perfect \textit{classical} strategy for a BCS game if and only if the underlying BCS is satisfiable.
It is interesting that there exist perfect entangled strategies for BCS games for some unsatisfiable BCSs.

%------------------------------------------------------------------------------%
\section{Mermin's quantum strategies}\label{sec:mermin}
%------------------------------------------------------------------------------%

Mermin~\cite{Mermin1990,Mermin1993} made a remarkable discovery about sets of observables with certain properties that has consequences for quantum strategies for BCS games%
\footnote{Mermin's original paper was written in the language of no-hidden-variables theorems, along the lines of the Kochen Specker Theorem; however, it discusses implications regarding Bell inequality violations, and these can be interpreted as non-local games where quantum strategies exist that outperform classical strategies.
The connection is made more explicit by Aravind~\cite{Aravind2002,Aravind2004}.}
that are unsatisfiable---in particular the following two games.
The left side of Fig.~\ref{fig:one} summarizes the BCS specified by the aforementioned system of equations~(\ref{eq:ms}).
We refer to this BCS as the \textit{magic square}.
%------------------------------------------------------------------------------%
\begin{figure}[ht!]
\begin{center}
\setlength{\unitlength}{1mm}
%------------------------------------------------------------------------------%
\begin{picture}(20,20)(0,0)
\linethickness{0.5pt}
\put(-2,19){$v_1$}
\put(8,19){$v_2$}
\put(18,19){$v_3$}
\put(-2,9){$v_4$}
\put(8,9){$v_5$}
\put(18,9){$v_6$}
\put(-2,-1){$v_7$}
\put(8,-1){$v_8$}
\put(18,-1){$v_{9}$}
\put(2,0){\line(1,0){5}}
\put(12,0){\line(1,0){5}}
\put(2,10){\line(1,0){5}}
\put(12,10){\line(1,0){5}}
\put(2,20){\line(1,0){5}}
\put(12,20){\line(1,0){5}}
\put(0,2){\line(0,1){5.5}}
\put(0,12){\line(0,1){5.5}}
\put(10,2){\line(0,1){5.5}}
\put(10,12){\line(0,1){5.5}}
\put(19.7,2){\line(0,1){5.5}}
\put(19.7,12){\line(0,1){5.5}}
\put(20.3,2){\line(0,1){5.5}}
\put(20.3,12){\line(0,1){5.5}}
\end{picture}
%------------------------------------------------------------------------------%
\hspace*{30mm}
\begin{picture}(40,35)(0,0)
\linethickness{0.5pt}
\put(2,19.7){\line(1,0){10.5}}
\put(16.8,19.7){\line(1,0){5.8}}
\put(26.8,19.7){\line(1,0){10.5}}
\put(2,20.3){\line(1,0){10.5}}
\put(16.8,20.3){\line(1,0){5.8}}
\put(26.8,20.3){\line(1,0){10.5}}
\put(19.5,32.9){\line(-1,-3){3.7}}
\put(14.3,18){\line(-1,-3){1.8}}
\put(11.3,9){\line(-1,-3){3.6}}
\put(20.6,32.9){\line(1,-3){3.7}}
\put(25.7,17.5){\line(1,-3){1.6}}
\put(28.6,9){\line(1,-3){3.6}}
\put(38.2,18.5){\line(-4,-3){8.5}}
\put(26.1,9.4){\line(-4,-3){4.1}}
\put(18.2,3.9){\line(-4,-3){8.7}}
\put(1.7,18.5){\line(4,-3){8.5}}
\put(14.7,9.0){\line(4,-3){3.8}}
\put(22.3,3.1){\line(4,-3){8.0}}
\put(18.2,34){$v_8$}
\put(-2,19){$v_7$}
\put(13,19){$v_5$}
\put(23,19){$v_1$}
\put(38,19){$v_9$}
\put(10.9,10){$v_4$}
\put(26.8,10){$v_2$}
\put(18.5,4){$v_3$}
\put(6,-4.5){$v_6$}
\put(30.8,-4.5){$v_{10}$}
\end{picture}
%------------------------------------------------------------------------------%
\end{center}
\caption{\small Structure of two BCSs: (a) magic square (left) and (b) magic pentagram (right).
Each straight line indicates a parity constraint on its variables of 0 for single lines, and 1 for double lines.
}\label{fig:one}
\end{figure}
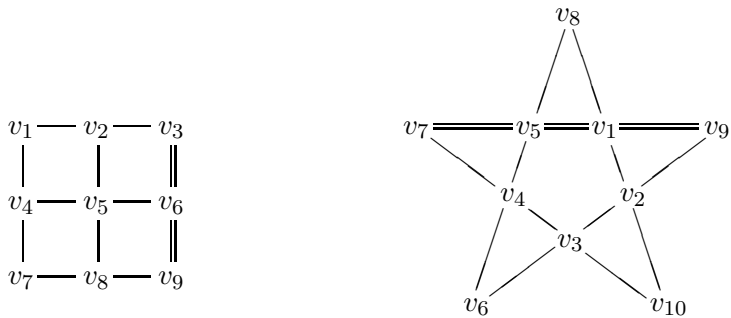
%------------------------------------------------------------------------------%
Similarly, the right side of Fig.~\ref{fig:one} summarizes another BCS consisting of ten variables and five constraints, where each constraint is related to the parity of four variables.
We refer to this BCS as the \textit{magic pentagram}.

To understand Mermin's strategies, we first define a \textit{quantum satisfying assignment} of a BCS as a relaxation of a classical satisfying assignment, in the following manner.
First translate each $\{0,1\}$-variable $v_j$ into a $\{+1,-1\}$-variable $V_j = (-1)^{v_j}$.
Then the parity of any sequence of variables is their product---and, in fact, every boolean function can be uniquely represented as a multilinear polynomial over $\mathbb{R}$ (e.g., for the binary OR-function, the polynomial is $(V_1 V_2 + V_1 + V_2 - 1)/2$).
Now we can define a quantum satisfying assignment as an assignment of finite-dimensional Hermitian operators 
$A_1, A_2, \dots, A_n$ to the variables $V_1, V_2, \dots, V_n$ (respectively) such that:
\begin{enumerate}[(a)]
\item
Each $A_j$ is a binary observable in that its eigenvalues are in $\{+1,-1\}$ (i.e.,$A_j^2 = I$).       
\item
All pairs of observables, $A_i$, $A_j$, that appear within the same constraint are commuting (i.e., 
they satisfy $A_i  A_j = A_j  A_i$).
\item
The observables \textit{satisfy} each constraint $c_s : \{+1,-1\}^k \rightarrow \{+1,-1\}$ that 
acts on variables $V_{i_1}, \dots, V_{i_k}$, in the sense that the multilinear polynomial equation $c_s(A_{i_1},\dots,A_{i_k}) = -I$ is satisfied.
\end{enumerate}
This is a relaxation of the standard ``classical" notion of a satisfying assignment (which corresponds to the case of one-dimensional observables).
Quantum satisfying assignments for the two BCSs in Figure~\ref{fig:one} are shown in Figure~\ref{fig:two}.
%------------------------------------------------------------------------------%
\begin{figure}[ht!]
\begin{center}
\setlength{\unitlength}{1mm}
%------------------------------------------------------------------------------%
\begin{picture}(20,20)(0,0)
\linethickness{0.5pt}
\put(-2,19){\scriptsize $ZI$}
\put(8,19){\scriptsize $IZ$}
\put(18,19){\scriptsize $ZZ$}
\put(-2,9){\scriptsize $IX$}
\put(8,9){\scriptsize $XI$}
\put(18,9){\scriptsize $XX$}
\put(-2,-1){\scriptsize $ZX$}
\put(8,-1){\scriptsize $XZ$}
\put(18,-1){\scriptsize $YY$}
\put(3,0){\line(1,0){4}}
\put(13,0){\line(1,0){4}}
\put(3,10){\line(1,0){4}}
\put(13,10){\line(1,0){4}}
\put(3,20){\line(1,0){4}}
\put(13,20){\line(1,0){4}}
\put(0,2){\line(0,1){5.5}}
\put(0,12){\line(0,1){5.5}}
\put(10,2){\line(0,1){5.5}}
\put(10,12){\line(0,1){5.5}}
\put(19.7,2){\line(0,1){5.5}}
\put(19.7,12){\line(0,1){5.5}}
\put(20.3,2){\line(0,1){5.5}}
\put(20.3,12){\line(0,1){5.5}}
\end{picture}
%------------------------------------------------------------------------------%
\hspace*{30mm}
%------------------------------------------------------------------------------%
\begin{picture}(40,35)(0,0)
\linethickness{0.5pt}
\put(4,19.7){\line(1,0){6.5}}
\put(18.7,19.7){\line(1,0){2.2}}
\put(27.8,19.7){\line(1,0){9.7}}
\put(4,20.3){\line(1,0){6.5}}
\put(18.7,20.3){\line(1,0){2.2}}
\put(27.8,20.3){\line(1,0){9.7}}
\put(19.5,32.9){\line(-1,-3){3.7}}
\put(14.3,18){\line(-1,-3){1.8}}
\put(11.3,9){\line(-1,-3){3.6}}
\put(20.6,32.9){\line(1,-3){3.7}}
\put(25.7,17.5){\line(1,-3){1.6}}
\put(28.6,9){\line(1,-3){3.6}}
\put(38.2,18.5){\line(-4,-3){8.1}}
\put(26.1,9.4){\line(-4,-3){4.1}}
\put(18.1,3.5){\line(-4,-3){7.6}}
\put(1.7,18.5){\line(4,-3){8.5}}
\put(14.7,9.0){\line(4,-3){3.6}}
\put(22.3,3.1){\line(4,-3){7.5}}
\put(17.2,34){\scriptsize $ZII$}
\put(-4,19){\scriptsize $XXZ$}
\put(11,19){\scriptsize $ZXX$}
\put(21,19){\scriptsize $ZZZ$}
\put(38,19){\scriptsize $XZX$}
\put(9.9,10){\scriptsize $IXI$}
\put(25.8,10){\scriptsize $IZI$}
\put(17.5,4){\scriptsize $XII$}
\put(5,-4.5){\scriptsize $IIX$}
\put(29.8,-4.5){\scriptsize $IIZ$}
\end{picture}
%------------------------------------------------------------------------------%
\end{center}
\caption{\small Quantum satisfying assignments for: (a) magic square (left) and (b) magic pentagram (right).
($X$, $Y$, and $Z$ are the usual $2\!\times\!2$ Pauli matrices, and juxtaposition means tensor product.)}\label{fig:two}
\end{figure}
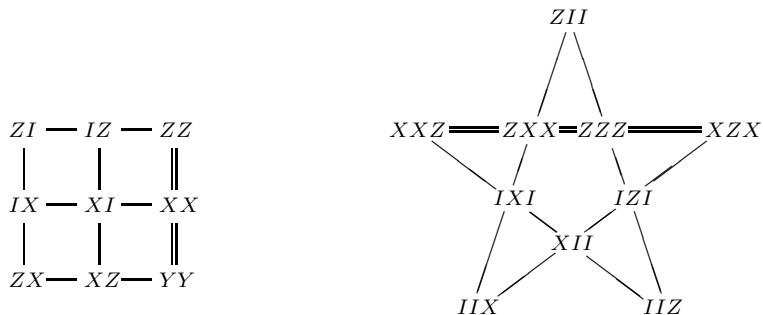
%------------------------------------------------------------------------------%

There is a construction (implicit in~\cite{Mermin1990} and explicit in~\cite{Aravind2004} for the magic square) that converts these quantum satisfying assignments into perfect strategies---and this is easily extendable to any quantum satisfying assignment of a BCS.
For completeness, we summarize the known construction.
The entanglement is of the form 
$\ket{\psi} = \frac{1}{\sqrt d}\sum_{j=1}^{d} \ket{j}\ket{j}$, 
where $d$ is the dimension of the observables.
Alice associates observables $A_1, A_2, \dots, A_n$ with the variables and Bob associates 
their transposes $A_{1}^{T}, A_{2}^{T} , \dots, A_{n}^{T}$ (with respect to the computational basis) with the variables.
On input $s$, Alice measures her observables that correspond to the variables in constraint $c_s$. 
At this point, it should be noted that this is a well-defined measurement since condition (b) implies that these observables are mutually commuting.
Also, on input $t$, Bob measures his observable $A_{t}^{T}$.
Condition (c) implies that Alice's output satisfies the constraint.
Finally, Alice and Bob give consistent values for variable $v_t$ because 
$\bra{\psi}A_t \otimes A_{t}^{T} \ket{\psi} 
= \bra{\psi} A_t \cdot A_{t} \otimes I \ket{\psi} 
= \bracket{\psi}{\psi} = 1$.

%------------------------------------------------------------------------------%
\section{General BCS games}\label{sec:questions}
%------------------------------------------------------------------------------%

A natural computational problem is: given a description of a BCS as input, determine whether or not it has a perfect entangled strategy.
A more general problem is to compute the maximum (or supremum) value of all entangled strategies.

For \textit{classical} strategies, the problem of determining whether or not a perfect strategy exists is NP-hard for general BCS games and in polynomial-time for parity BCS games (where the problem reduces to solving a system of linear equations in modulo 2 arithmetic).
For \textit{quantum} strategies, we are currently not aware of \textit{any} algorithm that determines whether or not an arbitrary parity BCS game has a perfect strategy (i.e., presently we do not even know that the problem is \textit{decidable}).

In Section~\ref{sec:characterization}, we prove a converse to the construction of entangled strategies from quantum satisfying assignments in Section~\ref{sec:mermin}.
Namely, we show that any perfect quantum strategy that uses countable-dimensional entanglement implies the existence of a quantum satisfying assignment.

It can be easily seen that not all BCS games have perfect quantum strategies, by this example 
\begin{equation}
v_1 \oplus v_2 = 0 
\ \ \ \ \ \ \ \ \ \ \ \ \ \ \ \ \ \ \ \ \ \ 
v_1 \oplus v_2 = 1.
\end{equation}
First note that no generality is lost if we assume that Alice returns only a value for $v_1$ (since the value of $v_2$ is then uniquely determined by the constraint).
It is not hard to see that such a game is equivalent to the so-called CHSH game~\cite{ClauserH+1969}, which is known to admit no perfect quantum strategy~\cite{Tsirelson1980} (even though the quantum success probability is higher than the classical success probability~\cite{ClauserH+1969}).
In Section~\ref{sec:gap}, we show how to derive upper bounds strictly below 1 on the entangled value of many parity BCSs.

%------------------------------------------------------------------------------%
\section{Characterization of perfect strategies in terms of observables}\label{sec:characterization}
%------------------------------------------------------------------------------%

%------------------------------------------------------------------------------%
\begin{theorem}\label{thm:characterization}
For any binary constraint system, if there exists a perfect quantum strategy for the corresponding BCS game that uses finite or countably-infinite dimensional entanglement, then it has a quantum satisfying assignment.
\end{theorem}
%------------------------------------------------------------------------------%

%------------------------------------------------------------------------------%
\begin{proof}
We start with an arbitrary binary constraint system with variables 
$v_1,v_2,\dots,v_n$ and constraints $c_1,c_2,\dots,c_m$.
Assume that there is a perfect entangled protocol for this system that uses entanglement 
\begin{equation}\label{eq:countable-entanglement}
\ket{\psi} = \sum_{i=1}^{\infty} \alpha_i \ket{\phi_i}\ket{\psi_i},
\end{equation}
where $\{\ket{\phi_1}, \ket{\phi_2}, \dots\}$ and $\{\ket{\psi_1}, \ket{\psi_2}, \dots\}$ are orthonormal sets, $\alpha_1, \alpha_2, \dots > 0$, and $\sum_{i=1}^{\infty} |\alpha_i|^2 = 1$.

%%%%%%%%%          NEW STUFF              %%%%%%%%%%%%%%%%%%%

We consider two separate cases for Alice's strategy. In the first case, she applies an arbitrary \textit{projective} measurement to the first register of $\ket{\psi}$.
%does not use any ancilla qubits for her measurements. 
%So, it can be assumed that Alice and Bob share $\ket{\psi} = \sum_{i=1}^{\infty} \alpha_i \ket{\phi_i}\ket{\psi_i}$, where,  $\forall i~\alpha_i > 0$. 
In the second case, Alice can apply an arbitrary POVM measurement to the first register of $\ket{\psi}$.

% uses private ancilla qubits for her measurement (Figure~\ref{fig:case2}). 
We will prove that quantum satisfying assignment exists in the first case. 
Then we will show that the second case can be reduced to first one, hence proving the theorem.

%%%%%%%%%%%%%%%%%%%%%%%%%%%%%%%%%%%%%%%%%
\medskip

\noindent\textbf{Case 1: Projective measurements for Alice.} For each $s \in \{1,2,\dots,n\}$, let $c_s$ be a constraint consisting of $r_s$ variables.
Therefore, the set of outcomes for Alice is $\{0,1\}^{r_s}$.
%without loss of generality%
%\footnote{Where we use the fact that we can transform any POVM measurement into Stinespring form.}, 
These can be associated with 
orthogonal projectors $\Pi_{a}^{s}$ ($a \in \{0,1\}^{r_s}$).
From these projectors, we can define the $r$ individual bits of the outcome as the binary observables
\begin{equation}
A_{s}^{(j)} = \sum_{a \in \{0,1\}^{r_s}} (-1)^{a_j} \Pi_a,
\end{equation}
for $j \in \{1,\dots,r_s\}$.
It is easy to check that $\{A_{s}^{(j)} : j \in \{1,\dots,r_s\}\}$ is a set of commuting 
binary observables.
We have defined a binary observable for Alice for each variable in the context of each 
constraint that includes it.
For example, in the case of the magic square (Eqns.~(\ref{eq:ms})), there is a binary observable $A_3^{(1)}$ for 
$v_7$ in the context of the third constraint and a binary observable $A_4^{(3)}$ for $v_7$ 
in the context of the fourth constraint.
We have not yet shown that $A_3^{(1)} = A_4^{(3)}$ (non-contextuality).

The measurements for Bob are (without loss of generality) 
binary observables $B_t$ for each variable $v_t$ ($t \in \{1,2,\dots,n\}$).

We need to show that the observables for Alice must be non-contextual: for each variable, Alice's 
observables for it are the same, regardless of the constraint that they arise from 
(for example, for the magic square game, $A_3^{(1)} = A_4^{(3)}$).
We shall use the following lemma.

%------------------------------------------------------------------------------%
%% lemma changed from B_1, B_2, C being observables to hermitian matrices between -I and I.

\begin{lemma}\label{lemma}
Let $-I \preceq C_1, C_2, B \preceq I$ be Hermitian matrices on some Hilbert space $\mathcal{H}$.
Let $\ket{\psi} \in \mathcal{H}\otimes\mathcal{H}$ be of the form
\begin{equation}
\ket{\psi} = \sum_{i=1}^{\infty} \alpha_i \ket{\phi_i}\ket{\psi_i},
\end{equation}
where $\{\ket{\phi_1},\ket{\phi_2},\dots\}$ and $\{\ket{\psi_1},\ket{\psi_2},\dots\}$ 
are orthonormal bases for $\mathcal{H}$, $\alpha_1,\alpha_2,\dots > 0$, and 
$\sum_{i=1}^{\infty} |\alpha_i|^2 = 1$.
Then, for the Hermitian matrices $ \{B, C_1, C_2\}$, if 
$\bra{\psi}B\otimes C_1\ket{\psi} = \bra{\psi}B\otimes C_2\ket{\psi} = 1$ 
then $C_1 = C_2$.
\end{lemma}
%------------------------------------------------------------------------------%

%------------------------------------------------------------------------------%
\begin{proof}[Proof of Lemma~\ref{lemma}]
Consider the vectors $w = B \otimes I \ket{\psi}$, $u_1 = I \otimes C_1 \ket{\psi}$, 
and $u_2 = I \otimes C_2 \ket{\psi}$.
These are vectors with length at most $1$ and we have $w \cdot u_1 = w \cdot u_2 = 1$, 
which implies that $u_1 = w = u_2$.
Therefore, 
\begin{eqnarray}
0 & = & I \otimes C_1 \ket{\psi} - I \otimes C_2 \ket{\psi} \\
& = & 
\left(I \otimes (C_1 - C_2)\right)
\left(\sum_{i=1}^{\infty} \alpha_i\ket{\phi_i}\ket{\psi_i}\right)\\
& = & \sum_{i=1}^{\infty} \alpha_i \ket{\phi_i}(C_1 - C_2)\ket{\psi_i},
\end{eqnarray}
which implies that $(C_1 - C_2)\ket{\phi_i} = 0$, for all $i \in \{1,2,\dots\}$.
This implies that $C_1 = C_2$, which completes the proof of the lemma.
\end{proof}
%------------------------------------------------------------------------------%

Returning to the proof of Theorem~\ref{thm:characterization}, let $t \in \{1,2,\dots,n\}$ and 
$A^{(j)}_{s}$ and $A^{(j^\prime)}_{s^\prime}$ be any two observables of Alice 
corresponding to the same variable $v_t$.
Since Alice's binary observables associated with constraint $c_s$ are commuting, 
we can assume that Alice begins her measurement process by measuring $A^{(j)}_{s}$, 
while Bob measures $B_t$.
Since these two measurements must yield the same outcome, we have 
$\bra{\psi}A^{(j)}_{s} \otimes B_t\ket{\psi} = 1$.
Similarly, $\bra{\psi}A^{(j^\prime)}_{s^\prime} \otimes B_t\ket{\psi} = 1$.
Therefore, applying Lemma~\ref{lemma}, we have $A^{(j)}_{s} = A^{(j^\prime)}_{s^\prime}$, 
which establishes that Alice's observables are non-contextual.

In addition to consistency between Alice and Bob, Alice's output bits must 
satisfy the constraint~$c_s$ (recall that $c_s$ can be expressed as a multilinear polynomial over $\mathbb{R}$).
That is,
\begin{equation}\label{eq:parity}
\bra{\psi}c_s(A^{(1)}_s,\dots,A^{(r_s)}_s) \otimes I \ket{\psi} = -1.
\end{equation}
By invoking Lemma~\ref{lemma} again, with 
$C_1 = -c_s(A^{(1)}_s,\dots,A^{(r_s)}_s)$, $C_2 = I$, 
$B=I$, we can deduce that 
$c_s(A^{(1)}_s,\dots,A^{(r_s)}_s) = -I$.

At this point, it is convenient to rename Alice's observables to $A_t$, for each 
$t \in \{1,2,\dots,n\}$ (which we can do because we proved they are non-contextual).
The observables associated with each constraint commute and their product has the 
required parity.

We will finally prove that a finite-dimensional set of observables must exist.
Since, for all $t \in \{1,2,\dots,n\}$, $\bra{\psi}A_t \otimes B_t\ket{\psi} = 1$, 
we have $A_t \otimes I \ket{\psi} = I \otimes B_t \ket{\psi}$, so
\begin{equation}\label{eq:schmidt}
\sum_{i=1}^{\infty} \alpha_i\left(A_t\ket{\phi_i}\right)\ket{\psi_i}
=
\sum_{i=1}^{\infty} \alpha_i\ket{\phi_i}\left(B_t\ket{\psi_i}\right).
\end{equation}
Both sides of Eq.~(\ref{eq:schmidt}) are Schmidt decompositions of the same quantum state.
Now we can use the fact that the Schmidt decomposition is unique up to a change of basis 
for the subspace associated with each distinct Schmidt coefficient.
Consider any Schmidt coefficient with multiplicity $d$ (each Schmidt coefficient appears 
with finite multiplicity).
Suppose, without loss of generality, that 
$\alpha_1 = \alpha_2 = \cdots = \alpha_d = \alpha$.
Then the span of $\left\{A_t\ket{\phi_i} : i \in \{1,2,\dots,d\}\right\}$ equals the 
span of $\left\{\ket{\phi_i} : i \in \{1,2,\dots,d\}\right\}$.
In other words, $A_t$ leaves the subspace spanned by 
$\left\{\ket{\phi_i} : i \in \{1,2,\dots,d\}\right\}$ fixed.
By similar reasoning, $B_t$ leaves the subspace spanned 
by $\left\{\ket{\psi_i} : i \in \{1,2,\dots,d\}\right\}$ fixed.
Therefore, there exist bases in which $A_t$ and $B_t$ have block decompositions of the form
\begin{equation}
A_t = \left(
\begin{array}{cccc}
A_{t}^{\prime} & 0 & 0 & \dots \\
0 & A_{t}^{\prime\prime} & 0 & \dots \\
0 & 0 & A_{t}^{\prime\prime\prime} & \dots \\
\vdots & \vdots & \vdots & \ddots
\end{array}
\right)
\ \ \ \ \ \ \ 
B_t = \left(
\begin{array}{cccc}
B_{t}^{\prime} & 0 & 0 & \dots \\
0 & B_{t}^{\prime\prime} & 0 & \dots \\
0 & 0 & B_{t}^{\prime\prime\prime} & \dots \\
\vdots & \vdots & \vdots & \ddots
\end{array}
\right)
\end{equation}
with one block for the subspace of each Schmidt coefficient.
We can take, say, the $d$-dimensional observables from the first block 
$\left\{A_{t}^{\prime} : t \in \{1,2,\dots,n\}\right\}$ as a quantum satisfying assignment 
(which changes the effective entanglement to a $d$-dimensional maximally entangled state).

%%%%%%%%%%    NEW STUFF   %%%%%%%%%%%%%%%%%%%%%%%
\medskip

\noindent\textbf{Case 2: POVM measurements for Alice.}
A POVM measurement can be expressed as a projective measurement in a larger Hilbert space that includes ancilliary qubits, as shown in Figure~\ref{fig:stinespring}.
Again we can define binary observables for $j^{th}$ variable in a constraint $s$ as in Case~1.
\begin{figure}
\[ 
\Qcircuit @C=1.1em @R=.5em {
    & \multigate{6}{ ~A^{(1)}_s ~} & \qw  & \multigate{6}{ ~A^{(2)}_s ~} & \qw & \cdots &  ~ &\multigate{6}{~A^{(r_s)}_s ~} & \qw \\
    \lstick{\begin{array}{r} \mbox{input} \\ \mbox{state} \end{array}}  
     & \ghost{~A^{(1)}_s~} & \qw & \ghost{~A^{(2)}_s~} & \qw & \cdots & ~ & \ghost{~A^{(r_s)}_s~} & \qw \\   
     & \ghost{~A^{(1)}_s~} & \qw & \ghost{~A^{(2)}_s~} & \qw & \cdots & ~ & \ghost{~A^{(r_s)}_s~} & \qw \\ 
     & & ~ \\
\lstick{\ket{0}}     & \ghost{~A^{(1)}_s~} & \qw & \ghost{~A^{(2)}_s~} & \qw & \cdots & ~ & \ghost{~A^{(r_s)}_s~} & \qw \\    
\lstick{\ket{0}}     & \ghost{~A^{(1)}_s~} & \qw & \ghost{~A^{(2)}_s~} & \qw & \cdots & ~ & \ghost{~A^{(r_s)}_s~} & \qw \\ 
\lstick{\ket{0}}     & \ghost{~A^{(1)}_s~} & \qw & \ghost{~A^{(2)}_s~} & \qw & \cdots & ~ & \ghost{~A^{(r_s)}_s~} & \qw 
}
\]
\caption{\small Alice's POVM measurement on receiving input $s$ expressed in Stinespring form (Case 2).}\label{fig:stinespring}
\end{figure}
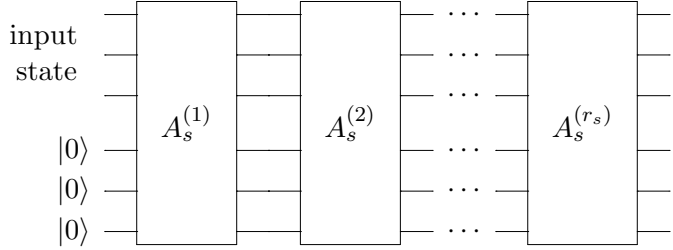
 
\begin{equation}
A_{s}^{(j)} = \sum_{a \in \{0,1\}^{r_s}} (-1)^{a_j} \Pi_a,
\end{equation}
these observables act on the larger Hilbert space $\mathcal{H}_s \otimes \mathcal{H}_p$. Here $\mathcal{H}_s$ ($\mathcal{H}_p$) represents the Hilbert space for the entangled (private) qubits. 
Like before, the $\{A_{s}^{(j)} : j \in \{1,\dots,r_s\}\}$ is a set of commuting binary observables.
Since these observables commute, without loss of generality, any of the corresponding variables can be measured first by Alice. 

We will focus on the first measurement done by Alice given some constraint. Let us suppress the superscript and subscript for brevity of notation. 
Say, Alice uses observable $A$ for the first measurement corresponding to variable $t$. This defines a projective measurement $(\Pi_0~=~\frac{A+I}{2}, \Pi_1~=~\frac{I-A}{2})$ on $\mathcal{H}_s \otimes \mathcal{H}_p$.

Suppose that the reduced entangled state on Alice's side is $\rho$. Then Alice's strategy is to apply the channel which adds the ancilla qubits to $\rho$ and then applies the measurement $(\Pi_0, \Pi_1)$. 
Using the Kraus operators of this channel, we can come up with \textit{equivalent} POVM elements $E_0, E_1$ acting on the Hilbert space $\mathcal{H}_s$. 
Here equivalent means, for all $i \in \{0,1\}$ and $\ket{\phi} \in \mathcal{H}_s$,
\begin{equation}
\label{equivalent}
\bra{\phi,00\dots 0} \Pi_i \ket{\phi,00\dots 0} = \bra{\phi} E_i \ket{\phi}.
\end{equation}
 
Similarly, Bob has POVM elements $(F_0, F_1)$ to measure variable $t$. Since their strategy is perfect, they always answer with same bit when asked for the variable $t$, 
which implies
\begin{equation}
\bra{\psi} E_0 \otimes F_0 \ket{\psi} + \bra{\psi} E_1 \otimes F_1 \ket{\psi} = 1.
\end{equation}
This can be simplified to 
\begin{equation}
\bra{\psi} (E_0 - E_1) \otimes (F_0 - F_1) \ket{\psi} = 1.
\end{equation}

Now we use the following lemma to prove that $(E_0, E_1)$ is actually a projective measurements (similarly $(F_0, F_1)$ is projective).

%------------------------------------------------------------------------------%
\begin{lemma}
\label{povm}
Let $\ket{\psi} \in \mathcal{H}_A \otimes \mathcal{H}_B$ be such that 
$\ket{\psi} = \sum_{i=1}^{\infty} \alpha_i \ket{\phi_i}\ket{\psi_i}$, where $\alpha_1, \alpha_2, \dots > 0$.  
If we have two POVM measurements, $(E_0, E_1)$ on $\mathcal{H}_A$ and $(F_0, F_1)$ on $\mathcal{H}_B$, such that 
\begin{equation}
\label{perfect}
\bra{\psi} (E_0 - E_1) \otimes (F_0 - F_1) \ket{\psi} = 1  
\end{equation}
then $(E_0, E_1)$ and $(F_0, F_1)$ are projective measurements.
\end{lemma}
%------------------------------------------------------------------------------%

%------------------------------------------------------------------------------%
\begin{proof}[Proof of Lemma~\ref{povm}]
We will prove that $(E_0, E_1)$ is a projective measurement. The proof for $(F_0, F_1)$ is the same.

Notice that $E_0$ and $E_1$ are simultaneously diagonalizable (they are both Hermitian and $E_0~+~E_1~=~I$). In the basis which diagonalizes them, 

\[ 
E_0 =  \left( \begin{array}{cccc}
\lambda_1 & ~ & ~ & ~ \\
~ & \lambda_2 & ~ & ~ \\
~ & ~ & \ddots & ~ \\
~ & ~ & ~ & \lambda_n
\end{array} \right)
~~\mbox{and}~~
E_1 =  \left( \begin{array}{cccc}
1-\lambda_1 & ~ & ~ & ~ \\
~ & 1-\lambda_2 & ~ & ~ \\
~ & ~ & \ddots & ~ \\
~ & ~ & ~ & 1-\lambda_n
\end{array} \right).
\]
This implies that $E_0$ and $E_1$ can be thought of as a probability distribution on 
$2^n$ projective measurements. 
For each $S \subseteq [n]$, define the projectors 
$\Pi^S_0 = \sum_{i \in S} \ket{i}\bra{i}$ and $\Pi^S_1 = I - \Pi^S_0$, 
and $p_S = \prod \limits_{i \in S} \lambda_i \prod \limits_{i \notin S} (1-\lambda_i)$.
Note that $\sum\limits_{S \subseteq [n]} p_S = 1$.
It is straightforward to verify that 
\begin{equation}\label{prob_dist}
E_0 = \sum_{S \subseteq [n]} p_S \Pi^S_0 
~~\mbox{and}~~ 
E_1 = \sum_{S \subseteq [n]} p_S \Pi^S_1.
\end{equation}
By Eqns.~(\ref{perfect}), (\ref{prob_dist}), and linearity,
\begin{equation}
\sum_{S \subseteq [n]} p_S ~ \bra{\psi} (\Pi^S_0 - \Pi^S_1) \otimes (F_0 - F_1)\ket{\psi} = 1.
\end{equation}

In the above equation, $p_S$'s sum up to $1$, and the term multiplied to them is at most $1$. By an averaging argument, for all $S$,
\begin{equation}
\label{eq:consistent3}
\bra{\psi} (\Pi^S_0 - \Pi^S_1) \otimes (F_0 - F_1)\ket{\psi} = 1.
\end{equation}

Using Lemma~\ref{lemma}, there can be at most one $p_S$ with non-zero probability. Hence $(E_0, E_1)$ is a projective measurement.

\end{proof}
%------------------------------------------------------------------------------%

Now we know that $(E_0, E_1)$ is a projective measurement. 
Also, using Eq.~(\ref{equivalent}), any eigenvector $\ket{\phi}$ of $E_0$ can be converted into an eigenvector $\ket{\phi, 00\cdots 0}$ for $\Pi_0$ with same eigenvalue.
Then, in the basis where eigenvalues of the form $\ket{\phi,00\cdots 0}$ are listed first, 

\begin{equation}
\Pi_0 = \left( \begin{array}{c|c} 
E_0 & 0~\cdots ~0 \\
\hline
\begin{array}{c}
0 \\
\vdots \\
0 \end{array}
&
\mbox{\LARGE{$M_0$}} 
\end{array}
\right)
~~~\mbox{and}~~~
\Pi_1 = \left( \begin{array}{c|c} 
E_1 & 0~\cdots ~0 \\
\hline
\begin{array}{c}
0 \\
\vdots \\
0 \end{array}
&
\mbox{\LARGE{$M_1$}}
\end{array} \right).
\end{equation}

%,\[ \left( \begin{array}{c} 0 \\ \vdots \\ 0 \end{array} \right) \] 

%\[ \begin{array}{cc} 
%\Pi_0 = \left( \begin{array}{cccc} 
%E_0 & 0 & \cdots & 0 \\
%0 & . & \cdots & . \\
%\vdots & \vdots & \ddots & \vdots \\
%0 &  . & \cdots & .
%\end{array} \right),
%&
%\Pi_1 = \left( \begin{array}{cccc} 
%E_1 & 0 & \cdots & 0 \\
%0 & . & \cdots & . \\
%\vdots & \vdots & \ddots & \vdots \\
%0 &  . & \cdots & .
%\end{array} \right).
%\end{array} \]

It is given that the observables $\Pi_0 - \Pi_1$ corresponding to different variables in the same context commute.
It follows that the observables $E_0 - E_1$ corresponding to different variables in the same context also commute.
Hence the proof for Case~2 follows from Case~1.

%%%%%%%%%%%%%%%%%%%%%%%%%%%%%%%%%%%%%%%%%
\end{proof}
%------------------------------------------------------------------------------%

%------------------------------------------------------------------------------%
\section{Proving gaps on the maximum quantum success probability}\label{sec:gap}
%------------------------------------------------------------------------------%

Theorem~\ref{thm:characterization} does not address strategies that employ more exotic kinds of infinite entanglement than expressed by Eq.~(\ref{eq:countable-entanglement}).
In particular, 
we have not ruled out the possibility that a binary constraint system exists for which there is no perfect strategy employing finite (or countably infinite) entanglement, but for which there is an infinite sequence of strategies, $\mathcal{P}_1, \mathcal{P}_2, \mathcal{P}_3, \dots$, where strategy $\mathcal{P}_d$ uses entanglement of the form
\begin{equation}\label{eq:max-ent}
\ket{\psi} = \frac{1}{\sqrt d}\sum_{j=1}^{d} \ket{j}\ket{j},
\end{equation}
and succeeds with probability $p_d < 1$ such that $\lim_{d \rightarrow \infty} p_d = 1$.
In this section, we show that this cannot happen for certain parity BCS games.

Consider the BCS illustrated in Figure~\ref{fig:4-line}, that we will refer to as the 
\textit{four-line} BCS (with each pair of lines intersecting)
%------------------------------------------------------------------------------%
\begin{figure}[ht!]
\begin{center}
\setlength{\unitlength}{0.85mm}
\begin{picture}(50,32)(3,6)
\linethickness{0.5pt}
\put(6.5,32.7){\line(2,-3){7.1}}
\put(16.4,17.9){\line(2,-3){7.2}}
\put(8,35){\line(1,0){14.2}}
\put(27.6,35){\line(1,0){14.7}}
\put(24.68,33){\line(0,-1){5.5}}
\put(25.32,33){\line(0,-1){5.5}}
\put(24.68,23){\line(0,-1){15.5}}
\put(25.32,23){\line(0,-1){15.5}}
\put(17.6,21.3){\line(2,1){5.8}}
\put(27.8,26.4){\line(2,1){14.9}}
\put(3.5,34.4){$v_1$}
\put(23,34.5){$v_2$}
\put(43,34.5){$v_3$}
\put(23.5,24.5){$v_4$}
\put(13.5,19.5){$v_5$}
\put(23.5,4.5){$v_6$}
\end{picture}
\end{center}
\caption{\small Structure of four-line BCS.}\label{fig:4-line}
\end{figure}
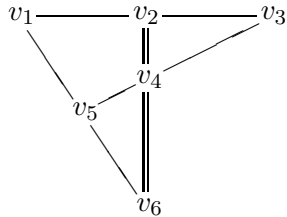
%------------------------------------------------------------------------------%
\newline which corresponds to the system of equations 
\begin{eqnarray}\label{eq:4-line}
v_1 \oplus v_2 \oplus v_3 & = & 0  \nonumber \\
v_3 \oplus v_4 \oplus v_5 & = & 0  \\
v_5 \oplus v_6 \oplus v_1 & = & 0  \nonumber \\
v_2 \oplus v_4 \oplus v_6 & = & 1. \nonumber
\end{eqnarray}

Speelman~\cite{Speelman2011} showed that this has no quantum satisfying assignment 
by some simple algebra.
Applying a sequence of substitutions among the observables yields:
\begin{eqnarray} 
A_1 A_2 A_3 & = & I 
\ \ \ \mbox{(corresponding to the first constraint)} \\
A_1 A_2 A_4  A_5 & = & I 
\ \ \ \mbox{(substituting $A_3 = A_4 A_5$ from the second constraint)} \\
A_1 A_2 A_4 A_6 A_1 & = & I 
\ \ \ \mbox{(substituting $A_5 = A_6 A_1$ from the third constraint)} \\
A_1 A_2 A_2 A_6 A_6 A_1 & = & -I 
\ \ \ \mbox{(substituting $A_4 = -A_2 A_6$ from the fourth constraint).} \label{eq:contradiction}
\end{eqnarray}
Note that Eq.~(\ref{eq:contradiction}), which we refer to as the \textit{final} equation of the process, simplifies to the contradiction $I = -I$ because each $A_{j}$ squares to $I$.

This \textit{substitution method} works for many other parity BCSs, where the
general methodology is to apply a sequence of substitutions to obtain a final equation (along the lines of Eq.~(\ref{eq:contradiction})) that has the property that it simplifies to $I = -I$ using $A^{2}_{j} = I$ (but not 
%using 
assuming any additional commutations).
For example, it is straightforward to use this method to show that the ``truncated pentagram" of Figure~\ref{fig:four} 
has no quantum satisfying assignment.
%------------------------------------------------------------------------------%
\begin{figure}[ht!]
\begin{center}
\setlength{\unitlength}{1mm}
\begin{picture}(40,25)(0,0)
\linethickness{0.5pt}
\put(2,19.7){\line(1,0){10.5}}
\put(16.8,19.7){\line(1,0){5.8}}
\put(26.8,19.7){\line(1,0){10.5}}
\put(2,20.3){\line(1,0){10.5}}
\put(16.8,20.3){\line(1,0){5.8}}
\put(26.8,20.3){\line(1,0){10.5}}
\put(14.3,18){\line(-1,-3){1.8}}
\put(11.3,9){\line(-1,-3){3.6}}
\put(25.7,17.5){\line(1,-3){1.6}}
\put(28.6,9){\line(1,-3){3.6}}
\put(38.2,18.5){\line(-4,-3){8.5}}
\put(26.1,9.4){\line(-4,-3){4.1}}
\put(18.2,3.9){\line(-4,-3){8.7}}
\put(1.7,18.5){\line(4,-3){8.5}}
\put(14.7,9.0){\line(4,-3){3.8}}
\put(22.3,3.1){\line(4,-3){8.0}}
\put(-2,19){$v_7$}
\put(13,19){$v_5$}
\put(23,19){$v_1$}
\put(38,19){$v_8$}
\put(10.9,10){$v_4$}
\put(26.8,10){$v_2$}
\put(18.5,4){$v_3$}
\put(6,-4.5){$v_6$}
\put(30.8,-4.5){$v_9$}
\end{picture}
\end{center}
\caption{\small Structure of truncated pentagram BCS.}\label{fig:four}
\end{figure}
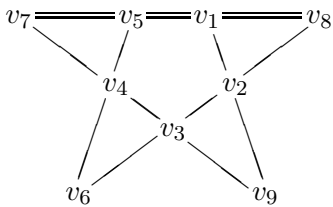
%------------------------------------------------------------------------------%

Can a sequence of strategies have success probability approaching 1 for any of these examples?
We prove the following theorem.

%------------------------------------------------------------------------------%
\begin{theorem}\label{thm:gap}
Whenever a parity BCS can be proven not to be quantum satisfiable by the substitution method, 
there exists a constant $\varepsilon > 0$ such that, for any strategy for the BCS game that uses finite entangement of the form of Eq.~(\ref{eq:max-ent}), the success probability is upper bounded by $1 - \varepsilon$ (where $\varepsilon$ is independent of the dimension of the entanglement).
\end{theorem}
%------------------------------------------------------------------------------%

%------------------------------------------------------------------------------%
\begin{proof}
As in the proof of Theorem~\ref{thm:characterization}, we can assume, without loss of generality, that Alice has 
commuting observables for each constraint and that their product is $\pm I$ in correspondence 
with the constraint.
This is because it is never advantageous for Alice to return bits that do not satisfy the 
constraint.
The difference here from the setting of the proof of Theorem~\ref{thm:characterization} is that, since Alice and Bob's bits do not have to be perfectly consistent,
Alice's observables can be contextual.
Thus, for each constraint $c_s$ and each variable $v_t$ within that constraint, there exists an observable $A_{t,s}$ that represents Alice's measurement for variable $v_t$ in the context of constraint $c_s$.
On Bob's side there remains one observable $B_t$ associated with each variable $v_t$.

We will show that, if the probability of Alice and Bob's bits being consistent is too high, then the substitution method still yields a contradiction.
Since the observables are contextual, the substituted variables are not actually eliminated; instead, a product of two versions (in different contexts) appears.
For example, the final equation for the four-lines BCS is 
$A^{\ }_{1} A^{\ }_{2} A^{\ }_{3} A^{\prime}_{3} A^{\ }_{4} A^{\prime}_{4} A^{\prime}_{2} A^{\ }_{6} A^{\ }_{5} A^{\prime}_{5} A^{\prime}_{6} A^{\prime}_{1} = -I$
(rather than Eq.~(\ref{eq:contradiction})), where $A_j$ and $A^{\prime}_{j}$ correspond to the two contexts of each variable $v_j$.
Let us suppose that the minimim consistency probability over all question pairs is $\cos\theta$ for some $\theta \ge 0$ (we will derive a lower bound on $\theta$).

For all question pairs, $(s_1,t)$ and $(s_2,t)$,
\begin{eqnarray}
(\bra{\psi}A_{t,s_1} \otimes I)
\cdot (I \otimes B_t \ket{\psi}) 
\ = \ 
\bra{\psi}A_{t,s_1} \otimes B_t \ket{\psi} & \ge &\cos\theta \\
(\bra{\psi}A_{t,s_2} \otimes I)
\cdot (I \otimes B_t \ket{\psi})
\ = \ 
\bra{\psi}A_{t,s_2} \otimes B_t \ket{\psi} & \ge & \cos\theta.
\end{eqnarray}
Our first observation is that,
\begin{equation}\label{eq:damage}
(\bra{\psi}A_{t,s_1} \otimes I)
\cdot (A_{t,s_2} \otimes I \ket{\psi}) 
= 
\bra{\psi}A_{t,s_1} A_{t,s_2} \otimes I\ket{\psi}
\ge \cos(2\theta).
\end{equation} 
This follows by considering the inner products among the three vectors 
$\bra{\psi}A_{t,s_1} \otimes I$, 
$\bra{\psi} I \otimes B_t$, 
$\bra{\psi}A_{t,s_2} \otimes I$ and noting that the extremal case is when they 
are co-planar.

Since, in general, $A_{t,s_1} \neq A_{t,s_2}$, the final equation from the substitution approach does 
not reduce to $I = -I$ (we cannot assume $A_{t,s_1} A_{t,s_2} = I$).
Nevertheless, Eq.~(\ref{eq:damage}) enables us to obtain a quantitative version of the contradiction, via the following lemma.

%------------------------------------------------------------------------------%
\begin{lemma}[approximate cancellation]\label{lemma:approx-sub}
Suppose that $\ket{\psi}$ is the maximally entangled state  
\begin{equation}
\ket{\psi} = \frac{1}{\sqrt d}\sum_{i=j}^{d} \ket{j}\ket{j}
\end{equation}
and let $A, B, B^{\prime}, C$ be binary observables such that 
\begin{equation}
\bra{\psi} A B B^{\prime} C \otimes I \ket{\psi} \ge \cos\Theta
\end{equation}
and
\begin{equation}
\bra{\psi}B^{\prime} B \otimes I\ket{\psi} \ge \cos(2\theta).
\end{equation}
Then 
\begin{equation}
\bra{\psi} A C\otimes I \ket{\psi} 
\ge \cos(\Theta + 2\theta).
\end{equation}
\end{lemma}
%------------------------------------------------------------------------------%

%------------------------------------------------------------------------------%
\begin{proof}[Proof of Lemma~\ref{lemma:approx-sub}]
Consider the vectors
\begin{eqnarray}
& & \bra{\psi}A B\otimes I \\
& & \bra{\psi}C B^{\prime} \otimes I \\
& & \bra{\psi}C B \otimes I.
\end{eqnarray}
The inner product between the first two vectors is at least $\cos\Theta$.
The inner product between the second and third vectors is 
\begin{eqnarray}
\bra{\psi}C B^{\prime} B C \otimes I \ket{\psi} 
& = &
\bra{\psi} B^{\prime} B \otimes C^T\!\cdot C^T \ket{\psi} \\
& = &
\bra{\psi} B^{\prime} B \otimes (C C)^T \ket{\psi} \\
& = &
\bra{\psi} B^{\prime} B \otimes I \ket{\psi} \\
& \ge & \cos(2\theta).
\end{eqnarray}
Therefore, the inner product between the first and third vector is at least 
$\cos(\Theta + 2\theta)$ because the extremal case is when the three vectors are co-planar.
Since $A B B C = A C$, this completes the proof of the lemma.
\end{proof}
%------------------------------------------------------------------------------%

Returning to the proof of Theorem~\ref{thm:gap},
we can apply Lemma~\ref{lemma:approx-sub} and Eq.~(\ref{eq:damage}) for each simplification step in the final equation arising from the substitution approach.
If there are $k$ such steps then we obtain
\begin{equation}
-1 = -\bra{\psi} I \otimes I \ket{\psi} \ge \cos(2k \theta),
\end{equation}
which is a contradiction unless $2k \theta \ge \pi$.
It follows that the minimum success probability over all questions is upper bounded by 
$\cos(\pi/2k)$, which is strictly below $1$.
\end{proof}
%------------------------------------------------------------------------------%

For the four-line BCS game (Figure~\ref{fig:4-line}, Eqns.~(\ref{eq:4-line})) the above approach implies that the minimum success probability over all questions is at most $(1+\cos(\pi/12))/2 = \cos^2(\pi/24) \approx 0.9830$.
From this we can immediately bound the value of the game by 
$1 - (1/12)\sin^2(\pi/24) \approx 0.9986$ (assuming a uniform distribution on questions).
We can obtain a better bound by considering averages of consistency probabilities rather than the minimum consistency probability, namely $\cos^2(\pi/24) \approx 0.9830$.
For comparison, the classical value of the four-line BCS game is $11/12 \approx 0.9167$.

%------------------------------------------------------------------------------%
\section{Conclusion}\label{sec:conclusion}
%------------------------------------------------------------------------------%

%------------------------------------------------------------------------------%
\subsection{Related work}\label{subsec:related}
%------------------------------------------------------------------------------%

There has been some interesting related work after the appearance of the first version of this article. Arkhipov~\cite{Arkhipov2012} studied parity BCS games where every variable appears 
exactly twice. He showed that any such game has perfect strategy if and only if a related \textit{dual} graph of the game is non-planar. The result combines elegant techniques with 
Kuratowski's theorem and our characterization of perfect strategies.

Very recently, Ji~\cite{Ji2013} showed that interesting examples like quantum chromatic number and Kochen-Specker sets can be described in the BCS game framework. He used special gadgets, 
called \textit{commutativity} gadgets,  to show reductions between various BCS's which preserve satisfiability using quantum assignments. Also, he showed that for all $k$, there exists a parity BCS game
which requires at least $k$ entangled qubits to play perfectly.

%------------------------------------------------------------------------------%
\subsection{Open questions}\label{subsec:open}
%------------------------------------------------------------------------------%

There are many questions left open by this work. We have a characterization of perfect strategies for BCS games. It shows that there always exists a perfect strategy using maximal entanglement 
if a perfect entangled strategy exist. Still, given a game, deciding whether it has a perfect strategy is open. 
Theorem~\ref{thm:gap} (which pertains to sequences of strategies whose success probability approaches $1$ in the limit) relies on the assumption that players use maximally entangled states. 
Is this assumption necessary? 

There are questions pertaining to the \textit{optimal} values of BCS games that admit no perfect strategies, such as computing them. 
Another question is whether there always exists an optimal strategy for a BCS game which uses maximally entangled states.

%What is the complexity of computing the \textit{optimal} value of BCS games? Can it be shown that there always exists an optimal strategy which uses maximally entangled states.

All of the above questions can be asked for general non-local games too. For the case of XOR games, the optimal value is given by a semidefinite program~\cite{CleveH+2004, Tsirelson1980}. 
This shows how to compute the optimal value of the game and that there always exist an optimal strategy which uses maximally entangled states~\cite{CleveH+2004}. It is also known for graph 
coloring games (like BCS games) that there always exists a perfect strategy using maximal entanglement (if a perfect entangled strategy exist)~\cite{CameronM+2007}. But whether this is true for general games that have perfect strategies remains open.

\section{Acknowledgments}
%------------------------------------------------------------------------------%

We are grateful for discussions about this project with many people, including Alex Arkhipov, Harry Buhrman, Sevag Gharibian, Tsuyoshi Ito, Kazuo Iwama, Zhengfeng Ji, Hirotada Kobayashi, Fran\c{c}ois Le Gall, Laura Mancinska (for pointing out an error 
in a previous version of this manuscript pertaining to the analysis of the case of POVM measurements), Oded Regev, Florian Speelman, Sarvagya Upadhyay, John Watrous, and Ronald de Wolf.
Some of this work took place while the first author was visiting Amsterdam's CWI in 2011.
This work is partially supported by Canada's NSERC and CIFAR.

%------------------------------------------------------------------------------%

\end{document}